\documentclass[sn-mathphys,Numbered]{sn-jnl}

\usepackage{graphicx}%
\usepackage{multirow}%
\usepackage{amsmath,amssymb,amsfonts}%
\usepackage{amsthm}%
\usepackage{mathrsfs}%
\usepackage[title]{appendix}%
\usepackage{xcolor}%
\usepackage{textcomp}%
\usepackage{manyfoot}%
\usepackage{booktabs}%
\usepackage{algorithm}%
\usepackage{algorithmicx}%
\usepackage{algpseudocode}%
\usepackage{listings}%

\theoremstyle{thmstyleone}%
\newtheorem{theorem}{Theorem}
\newtheorem{proposition}[theorem]{Proposition}%
\newtheorem{lemma}[theorem]{Lemma}

\theoremstyle{thmstyletwo}%

\theoremstyle{thmstylethree}%
\newtheorem{definition}{Definition}%

\begin{document}

\title{Atom graph, partial Boolean algebra and quantum contextuality}

\author[1]{\fnm{Songyi} \sur{Liu}}\email{liusongyi@buaa.edu.cn}

\author*[1]{\fnm{Yongjun} \sur{Wang}}\email{wangyj@buaa.edu.cn}

\author[1]{\fnm{Baoshan} \sur{Wang}}\email{bwang@buaa.edu.cn}

\author[1]{\fnm{Jian} \sur{Yan}}\email{jianyanmath@buaa.edu.cn}

\author[1]{\fnm{Heng} \sur{Zhou}}\email{zhouheng@buaa.edu.cn}

\affil*[1]{\orgdiv{School of Mathematical Sciences}, \orgname{Beihang University}, \orgaddress{ \city{Beijing}, \postcode{100191}, \country{China}}}

\abstract{Partial Boolean algebra underlies the quantum logic as an important tool for quantum contextuality. We propose the notion atom graphs to reveal the graph structure of partial Boolean algebra for finite dimensional quantum systems by proving that (i) the partial Boolean algebras for quantum systems are determined by their atom graphs; (ii) the states on atom graphs can be extended uniquely to the partial Boolean algebras, and (iii) each exclusivity graph is an induced graph of an atom graph. (i) and (ii) show that the finite dimensional quantum systems are uniquely determined by their atom graphs. which proves the reasonability of graphs as the models of quantum experiments. (iii) establishes a connection between atom graphs and exclusivity graphs, and introduces a method to express the exclusivity experiments more precisely. We also present a general and parametric description for Kochen-Specker theorem based on graphs, which gives a type of non-contextuality inequality for KS contextuality.}




\maketitle

\section*{Declarations}

\bmhead{Competing interests}
The authors have no relevant financial or non-financial interests to disclose.

\section{Introduction}

Quantum theory provides potential capabilities for information processing. The investigation of fundamental features of quantum systems has become a significant issue. All the non-classical features of quantum systems, such as non-locality \citep{Bell1964On,Popescu1994Quantum}, negativity \citep{Wootters1986The} and Kochen-Specker contextuality \citep{Kochen1967The}, can be generalized by quantum contextuality, which is divided into state-dependent contextuality and state-independent contextuality \citep{Adan2023Kochen}. It was shown that contextuality supplies a critical resource for quantum computation \citep{Mark2014Contextuality}.

Partial Boolean algebra is a powerful tool for quantum contextuality, which was used by Kochen and Specker (1967) to examine the problem of hidden variables in quantum mechanics \citep{Kochen1967The}, and has achieved great development for logic of quantum mechanics \citep{Isham1998Topos,Van2012Noncommutativity,Abramsky2020The}. A quantum system consists of a measurement scenario and a quantum state. The measurement scenario introduces contexts and the quantum state supplies super-classical probability distributions, which cause the contextuality together. A measurement scenario forms a partial Boolean algebra, and the quantum states are described by the probability distributions.

In this paper, the partial Boolean algebras are shown to be linked with the exclusivity graphs, which are utilized to depict quantum probabilities and non-contextuality inequalities (NC inequalities) \citep{Adan2010Contextuality,Adan2014Graph,Mark2014Contextuality}. We explore the features of partial Boolean algebras for quantum systems, and get some results. Firstly, we propose the atom graphs, and expose the graph structures of finite dimensional $epBA$, that is, partial Boolean algebra satisfying logical exclusivity principle (LEP). Therefore, a finite dimensional quantum system is uniquely determined by graph with probability distributions on them. Therefore, the utilization of graphs to be the models of quantum systems is proved reasonable. Secondly, we present a method to extend every exclusivity graph to an atom graph, which establishes a connection between partial Boolean algebra and exclusivity graphs. Finally, we introduce a general and parametric description for Kochen-Specker theorem based on graphs, which gives a type of NC inequality for KS contextuality.

In the next section 2, the concept of partial Boolean algebra is introduced. Section 3 defines atom graphs, and shows the graph structures of finite dimensional quantum systems with two theorems. In Section 4, it is proved that each finite graph is the induced subgraph of atom graph. Section 5 obtains a parametric description of KS contextuality. Finally, in Section 6, we summarize our work.

\section{Partial Boolean algebra}
\subsection{Basic concepts}
Partial Boolean algebra is generalization of Boolean algebra. Some concepts defined below are from \cite{Van2012Noncommutativity,Abramsky2020The}.

\begin{definition}[partial Boolean algebra]
    If $B$ is a set with
    \begin{itemize}
    \item a reflexive and symmetric binary relation $\odot\subseteq B\times B$,
    \item a (total) unary operation $\lnot:\ B\rightarrow B$,
    \item two (partial) binary operations $\land,\ \lor:\ \odot\rightarrow B$,
    \item elements $0,1 \in B$,
    \end{itemize}
    satisfying that for every subset $S\subseteq B$ such that $\forall a,\ b\in S,\ a\odot b$, there exists a Boolean subalgebra $C\subseteq B$ determined by $(C,\land,\lor,\lnot,0,1)$ and $S\subseteq C$, then $B$ is called a \textbf{partial Boolean algebra}, written by $(B,\odot)$, or $(B,\odot;\land,\lor,\lnot,0,1)$ for details.\par
    We use $pBA$ to denote the collection of all partial Boolean algebras.
\end{definition}

The abbreviation $pBA$ is adopted from \cite{Abramsky2020The}, where $pBA$ represents the category of partial Boolean algebras. The relation $\odot$ represents the compatibility. $a\odot b$ if and only if $a,b$ belong to a Boolean subalgebra. Therefore, a Boolean subalgebra of $B$ is called a context, and a maximal Boolean subalgebra is called a maximal context.\par

A partial Boolean algebra $B$ can be seen as overlapped Boolean algebras. More specifically, $B$ is a colimit of its total subalgebras in the category of partial Boolean algebras \citep{Van2012Noncommutativity}. For elements $a,b\in B$, we write $a\leq b$ to mean that $a\odot b$ and $a\land b=a$.

\begin{definition}
     Let $B\in pBA$, $a\in B$ and $a\neq 0$. $a$ is called an \textbf{atom} of $B$ if for each $x\in B$, $x\leq a$ implies $x=0$ or $x=a$. Use $A(B)$ to denote the atoms set of $B$.
\end{definition}

\begin{definition}
 Let $B\in pBA$. \par
 $B$ is said to be \textbf{atomic} if for each $x\in B$ and $x\neq 0$, there is $a\in A(B)$ such that $a\leq x$. \par
 $B$ is said to be \textbf{complete} if for each subset $S\subseteq B$ whose elements are pairwise compatible, $\bigvee S$ exists.\par
\end{definition}

We mainly concern finite systems. If an observable possesses infinite spectrum of eigenvalues, only a finite number of eigenspaces will be considered in practical fields such as quantum computation. Therefore, we define

\begin{definition}
If $B\in pBA$ and $B$ only contains finite Boolean subalgebras, then $B$ is said to be \textbf{finite dimensional}. Define $d(B):=\max\limits_{C}|A(C)|$ as the \textbf{dimension} of $B$  where $C$ is the Boolean subalgebras of $B$ and $|A(C)|$ is the size of $A(C)$.
\end{definition}

Each finite Boolean algebra is atomic and complete, so each finite dimensional partial Boolean algebra is also atomic and complete.

Abramsky et al. extends the exclusivity principle from quantum states to partial Boolean algebras to get closer to a quantum-realisable model\citep{Abramsky2020The}. The relevant definition is shown below.

\begin{definition}
 Let $B\in pBA$. \par
 $a,b\in B$ are said to be exclusive, written $a\bot b$, if there exists an element $c\in B$ such that $a\leq c$ and $b\leq\neg c$. \par
 $B$ is said to satisfy \textbf{Logical Exclusivity Principle (LEP)} or to be \textbf{exclusive} if $\bot\subseteq\odot$.\par
 Use $epBA$ to denote the collection of exclusive partial Boolean algebras, $acepBA$ to denote the atomic, complete and exclusive partial Boolean algebras.
\end{definition}

The abbreviation  $epBA$ is from \cite{Abramsky2020The}.

\begin{definition}
Let $B$ be a partial Boolean algebra. A Boolean subalgebra $C\subseteq B$ is called a \textbf{maximal Boolean subalgebra} of $B$ if for each Boolean subalgebra $D\subseteq B$,  $D\supseteq C$ implies $D=C$.
\end{definition}

\begin{definition}
If $B\in pBA$,  then a \textbf{state} on $B$ is defined by a map $p:B\rightarrow[0,\ 1]$ such that
 \begin{itemize}
    \item $p(0)=0$.
    \item $p(\neg x)=1-p(x)$.
    \item for all $x,y\in B$ with $x\odot y$,$\ p(x\lor y)+p(x\land y)=p(x)+p(y)$.
 \end{itemize}
A state is called a \textbf{0-1 state} if its range is $\{0,1\}$. Use $s(B)$ to denote the states set on $B$.\label{def_state}
\end{definition}

States are utilized to depict probability distributions of systems. A 0-1 state is a homomorphism from a partial Boolean algebra to \{0,1\}, that is, a truth-values assignment.

\subsection{Quantum system}
In this subsection, we show how to describe quantum systems using partial Boolean algebras.

Quantum logic was proposed by Birkhoff and Von Neumann (1936) to describe the property deduction in quantum physics \citep{Birkhoff1936The}. Quantum states are depicted by a Hilbert space $\mathcal{H}$. A proposition like $\hat{A}\in\Delta$ is depicted by a projector $\hat{P}$ on $\mathcal{H}$, where $\hat{A}$ is a bounded self-adjoint operator on $\mathcal{H}$ representing a physical quantity, and $\Delta$ is a Borel set of $\mathbb{R}$. Therefore, properties in a quantum system compose a set of projectors $\mathcal{P}(\mathcal{H})$. If $\hat{P}_1,\ \hat{P}_2$ are projectors onto closed linear subspaces $S_1,\ S_2$, $\hat{P}_1\land\hat{P}_2$ is defined to be the projector onto $S_1\cap S_2$, and $\lnot\hat{P}_1$ is defined to be the projector onto $S_1^{\bot}$. Then $\hat{P}_1\lor\hat{P}_2=\lnot(\lnot\hat{P}_1\land\lnot\hat{P}_2)$. One can prove that $\mathcal{P}(\mathcal{H})$ is an orthocomplemented modular lattice, called property lattice or standard quantum logic.

Property lattice $\mathcal{P}(\mathcal{H})$ has several disadvantages such as not satisfying the distributive law \citep{Doering2010Topos}. In research of contextuality, partial Boolean algebra performs better than orthocomplemented modular lattice. Therefore, we let $\mathcal{P}(\mathcal{H})$ be a partial Boolean algebra, which means operations between the noncommutative projectors are not allowed.

For details, all the projectors on $\mathcal{H}$ constitute the set $\mathcal{P}(\mathcal{H})$. Define binary relation $\hat{P}_1\odot\hat{P}_2$ by $\hat{P}_1\hat{P}_2=\hat{P}_2\hat{P}_1$. $\hat{P}_1\land\hat{P}_2$ is defined to be $\hat{P}_1\hat{P}_2$ only if $\hat{P}_1\odot\hat{P}_2$,. Definition of $\lnot\hat{P}_1$ is unchanged. Then we have $\hat{P}_1\lor\hat{P}_2=\lnot(\lnot\hat{P}_1\land\lnot\hat{P}_2)=\hat{P}_1+\hat{P}_2$. Because pairwise commeasurable projectors generate a Boolean algebra, $\mathcal{P}(\mathcal{H})=(\mathcal{P}(\mathcal{H}),\ \odot;\ \land,\ \lor,\ \lnot,\ \hat{0},\ \hat{1})$ is a partial Boolean algebra, where $\hat{0}$ is the zero projector, and $\hat{1}$ is the projector onto $\mathcal{H}$.

 We don't need to consider all the observables, that is, bounded self-adjoint operators on $\mathcal{H}$. In that case, we will get a partial algebra rather than a partial Boolean algebra. Because each bounded self-adjoint operator has spectral decomposition, and all the propositions about observables can be described by projectors, $\mathcal{P}(\mathcal{H})$ is powerful enough for us.

Easy to see $\mathcal{P}(\mathcal{H})$ is atomic and complete. The atoms of $\mathcal{P}(\mathcal{H})$ are the total rank-1 projectors. And each finite quantum system, that is, finite partial Boolean subalgebra of $\mathcal{P}(\mathcal{H})$, is naturally atomic and complete.

Consider four projectors on a qubit (2-dimensional Hilbert space), $\hat{P}_0=|0\rangle\langle0|,\ \hat{P}_1=|1\rangle\langle1|,\ \hat{P}_+=|+\rangle\langle+|,\ \hat{P}_-=|-\rangle\langle-|$, which generate the partial Boolean algebra in Fig.\ref{fig1}.

 \begin{figure}[h]
    \centering
    \includegraphics[width=0.4\linewidth]{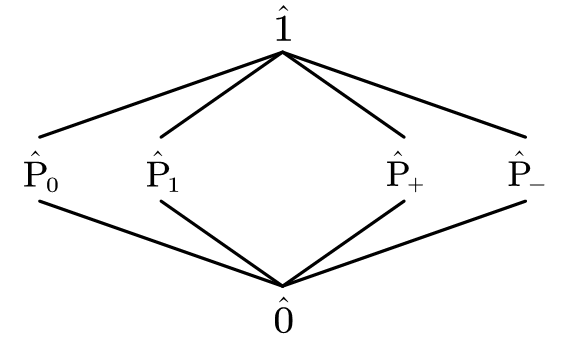}
    \caption{Partial Boolean algebra generated by $\hat{P}_0$, $\hat{P}_1$, $\hat{P}_+$, $\hat{P}_-$ ($\hat{P}_0\odot\hat{P}_1$, $\hat{P}_+\odot\hat{P}_-$).}\label{fig1}
 \end{figure}

We can also draw ``overlapped" partial Boolean algebras. Five rank-1 projectors on 3-dimensional Hilbert space as Fig.\ref{fig2} generate the partial Boolean algebra shown in Fig.\ref{fig3}.

\begin{figure}[h]
    \centering
    \includegraphics[width=0.35\linewidth]{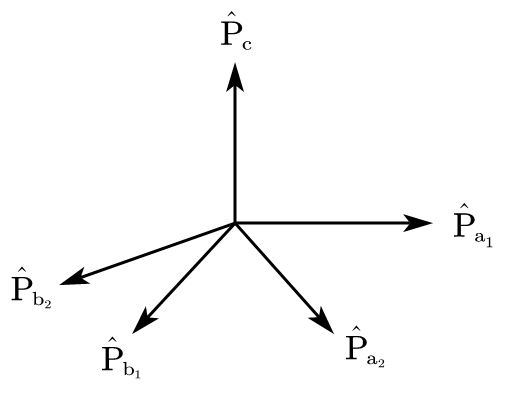}
    \caption{Five 3-dimensional rank-1 projectors. $\hat{P}_c,\hat{P}_{a_1},\hat{P}_{b_1}$ are pairwise orthogonal and $\hat{P}_c, \hat{P}_{a_2},\hat{P}_{b_2}$ are pairwise orthogonal}\label{fig2}
 \end{figure}

 \begin{figure}[h]
    \centering
    \includegraphics[width=0.5\linewidth]{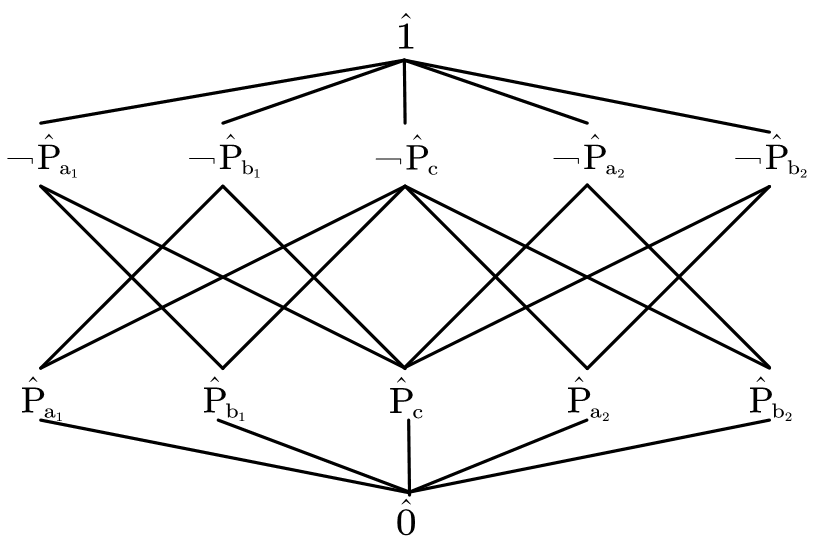}
    \caption{Partial Boolean algebra generated by $\hat{P}_c$, $\hat{P}_{a_1}$, $\hat{P}_{b_1}$, $\hat{P}_{a_2}$,$\hat{P}_{b_2}$}\label{fig3}
 \end{figure}

 Every experiment of quantum physics chooses a finite partial subalgebra of $\mathcal{P}(\mathcal{H})$ as its measurement scenario. The measurement scenario of CHSH experiment for Bell inequality is a partial Boolean algebra with 16 atoms \citep{Clauser1969Proposed} (4 observables introduce 16 elementary events), and the KCBS experiment for NC inequality is generated by 5 atoms \citep{Alexander2008Simple}.

 If two projectors $\hat{P}_1,\ \hat{P}_2$ are exclusive, which means there is a projector $\hat{P}$ such that $\hat{P}_1\leq\hat{P}$ and $\hat{P}_2\leq\lnot\hat{P}$, then $\hat{P}_1,\ \hat{P}_2$ are orthogonal, so they are commutative. Therefore, $\mathcal{P}(\mathcal{H})\in epBA$.

 To sum up, $\mathcal{P}(\mathcal{H})\in epBA$. Easy to see that each partial Boolean subalgebra of $\mathcal{P}(\mathcal{H})$ is also exclusive. Any ``quantum system" on $\mathcal{H}$ can be treated as partial Boolean subalgebra of $P(\mathcal{H})$.

 \begin{definition}
A \textbf{quantum system} is defined by a partial Boolean subalgebra of $P(\mathcal{H})$ for some Hilbert space $\mathcal{H}$. We use $QS$ to denote the collection of all quantum systems.
\end{definition}

We have $QS\subseteq epBA\subseteq pBA$. The axiomatization of quantum systems may need more extra properties (such as \cite{Popescu1994Quantum}).

\section{Atom graph}\label{sec3}
In this section, we define the atom graphs and prove several theorems which expose the graph structures of quantum systems. Unless otherwise specified, the graphs in this paper are simple and undirected.

\subsection{Graph Structure Theorem of acepBA}

If $B$ is an atomic and complete Boolean algebra, then $B$ is determined by its set of atoms, in other words, $B$ is isomorphic to the algebra of the power set of atoms. We generalize the conclusion to $acepBA$.

\begin{definition}
If $B\in pBA$, the \textbf{atom graph} of $B$, written $AG(B)$, is defined by a graph with vertex set $A(B)$ such that $a_1,a_2\in A(B)$ are adjacent iff $a_1\odot a_2$ and $a_1\neq a_2$ .
\end{definition}

For example, the atom graph of partial Boolean algebra in Fig.\ref{fig3} is shown in Fig.\ref{fig4}.

\begin{figure}[H]
    \centering
    \includegraphics[width=0.3\linewidth]{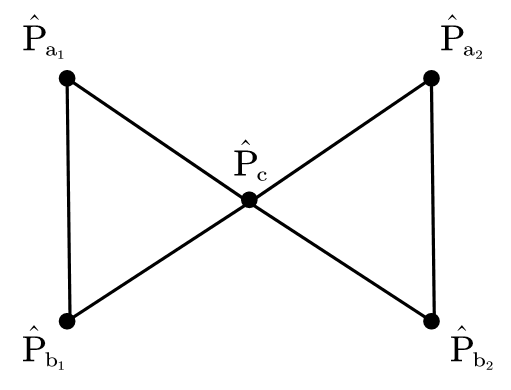}
    \caption{The atom graph of partial Boolean algebra in Fig.\ref{fig3}}\label{fig4}
 \end{figure}

Now we prove that the structure of an $acepBA$ is uniquely determined by its atom graph.

\begin{theorem}\label{thm1}
 If $B_1,B_2\in acepBA$. then $B_1\cong B_2$ iff $AG(B_1)\cong AG(B_2)$.
\end{theorem}
\begin{proof}
If $B_1\cong B_2$, $AG(B_1)\cong AG(B_2)$ obviously from the relevant definitions.\par

Conversely, if $g:A(B_1)\to A(B_2)$ is an isomorphism between $AG(B_1)$ and $AG(B_2)$, in other words, $a_1,a_2$ are adjacent iff $g(a_1),g(a_2)$ are adjacent, we define a map from $B_1$ to $B_2$ as follow.
\begin{equation}
\begin{split}
f:B_1&\to B_2\\
0 &\mapsto 0\\
b=\bigvee A_1&\mapsto\bigvee g(A_1).\ (A_1\subseteq A(B_1))
\nonumber
\end{split}
\end{equation}

To finish the proof, the lemma below is necessary.

If $B\in acepBA$ and $C\subseteq B$ is a maximal Boolean subalgebra, it is easy to proved that $A(C)\subseteq A(B)$. Furthermore, if $C_1,C_2$ are two maximal Boolean subalgebras of $B$, $A_1\subseteq A(C_1)$, $A_2\subseteq A(C_2)$, $A'_1=A(C_1)-A_1$ and $A'_2=A(C_2)-A_2$, then $\bigvee A_1=\bigvee A_2$ iff $A'_1\cup A_2=A(D_1)$ and $A_1\cup A'_2=A(D_2)$, where $D_1,D_2$ are maximal Boolean subalgebras of $B$.

Necessity: If $\bigvee A_1=\bigvee A_2=b\in B$, then $\bigvee A'_1=\bigvee A'_2=\neg b$. For any $a'_1\in A'_1$, $a'_1\leq\neg b$, and for any $a_2\in A_2$, $a_2\leq b$, so $a'_1\bot a_2$. Since $B$ is exclusive, $a'_1\odot a_2$. Therefore $A'_1\cup A_2$ is contained in a maximal Boolean subalgebra $D_1$, and $A'_1\cup A_2\subseteq A(D_1)$ due to $A'_1\cup A_2$ are atoms of $B$. Because $\bigvee (A'_1\cup A_2)=1$, $A'_1\cup A_2=A(D_1)$. Identically, $A_1\cup A'_2=A(D_2)$.\par

Sufficiency: If $A'_1\cup A_2=A(D_1)$ and $A_1\cup A'_2=A(D_2)$, we firstly prove that $A'_1\cap A_2=\emptyset$. Suppose there exists $a\in A'_1\cap A_2$, then $a\notin A_1$ and $a\notin A'_2$. On the other side, for all $a_1\in A_1$ and $a'_2\in A'_2$, $a\odot a_1$ and $a\odot a'_2$. Therefore $\{a\}\cup A_1 \cup A'_2$ is contained in a Boolean subalgebra, which contradicts that $D_2$ is maximal, so $A'_1\cap A_2=\emptyset$. Therefore $\bigvee A_1=\neg(\bigvee A'_1)=\bigvee A_2$.

Now we can prove that the map $f:B_1\to B_2$ is an isomorphism. The $C,C_1$ and $C_2$ below are all maximal Boolean subalgebras of $B_1$.

$f(b)=\bigvee g(A_1)$ exists because $B_2$ is complete. Suppose $A_1\subseteq A(C_1)$ and $A_1\subseteq A(C_2)$. If $b=\bigvee A_1=\bigvee A_2$, then $A'_1\cup A_2=A(D_1)$ and $A_1\cup A'_2=A(D_2)$ due to the proportion $1$ in the lemma above, so $A'_1\cup A_2$ and $A_1\cup A'_2$ are both maximal cliques of $AG(B_1)$. Suppose that $\bigvee g(A_1)\neq\bigvee g(A_2)$. Because of the proportion $2$ in the lemma above, one of $g(A'_1\cup A_2)$ and $g(A_1\cup A'_2)$ is not a maximal clique of $AG(B_2)$, which contradicts that $g$ is an isomorphism between $AG(B_1)$ and $AG(B_2)$. Therefore, $f$ is well-defined. \par

If $b_1=\bigvee A^1_{1}\in B_1$, $b_2=\bigvee A^2_{1}\in B_1$ and $b_1\neq b_2$, then $f(b_1)=\bigvee g(A^1_1)$, $f(b_2)=\bigvee g(A^2_1)$. Similarly, $f(b_1)\neq f(b_2)$ because of the lemma above and the isomorphism $g$, so $f$ is injective. For any $b=\bigvee A\in B_2$, $f(\bigvee g^{-1}(A))=b$, so $f$ is surjective. Therefore, $f$ is a bijection.\par

Finally, $f(0)=0$. For $b=\bigvee A\in B_1$, suppose that $A\subseteq C$. $f(\neg b)=f(\neg(\bigvee A))=f(\bigvee A')=\bigvee g(A')=\neg\bigvee g(A)=\neg f(b)$. If $b_1, b_2\in B_1$ and $b_1\odot b_2$, then let $b_1=\bigvee A_1,b_2=\bigvee A_2\in C$, so $A_1,A_2\subseteq A(C)$, $f(b_1),f(b_2)\in f(C)$ and $f(b_1)\odot f(b_2)$. Furthermore, $f(b_1\lor b_2)=f(\bigvee A_1\lor\bigvee A_2)=f(\bigvee(A_1\cup A_2))=\bigvee g(A_1\cup A_2)=\bigvee g(A_1)\lor\bigvee g(A_2)=f(b_1)\lor f(b_2)$. Therefore, $f$ is a homomorphism.\par

In conclusion, $f$ is an isomorphism between $B_1$ and $B_2$.
\end{proof}

Therefore, if $Q_1,Q_2$ are atomic and complete $QS$, then $Q_1\cong Q_2$ iff $AG(Q_1)\cong AG(Q_2)$, which exposes the graph structure of measurement scenarios of atomic and complete quantum systems.

In 2020, Abramsky and Barbosa proposed a tool to extend the compatibility relation of a partial Boolean algebra
\citep{Abramsky2020The}, that is, $B\rightarrow B[\circledcirc]$ ($\circledcirc$ is a binary relation of $B$). Theorem \ref{thm1} implies that, for an $acepBA$, the extension of compatibility relation is equivalent to the increasing of edges of its atom graph.

\subsection{Extension Theorem of the states on atom graphs}

We have defined states on partial Boolean algebras. For graphs, we have definition below.

\begin{definition}
If $G$ is a simple graph whose cliques have finite sizes, a \textbf{state} on $G$ is defined by a map $p:V(G)\rightarrow[0,\ 1]$ such that for each maximal clique $C$ of $G$, $\sum_{v\in C}p(v)=1$. Use $s(G)$ to denote the states set on $G$.
\end{definition}

Abramsky and Barbosa pointed out that there is a one-to-one correspondence between the states on a finite Boolean algebra and the probability distributions on the atoms \citep{Abramsky2020The}. We generalize the conclusion to finite dimensional $epBA$.

\begin{theorem}\label{thm2}
If $B$ is a finite dimensional $epBA$, then $s(B)\cong s(AG(B))$.
\end{theorem}
\begin{proof}
If $p\in s(B)$, obviously $p|_{A(B)}\in s(AG(B))$. Define $f:s(B)\to s(AG(B))$ as the restriction map below. We prove that $f$ is a bijection.

\begin{equation}
\begin{split}
f:s(B)&\to s(AG(B))\\
p&\mapsto p|_{A(B)}
\nonumber
\end{split}
\end{equation}

If $p_1,p_2\in s(B)$ and $p_1\neq p_2$, then $f(p_1)\neq f(p_2)$. Otherwise, suppose that $f(p_1)=f(p_2)=p_1|_{A(B)}=p_2|_{A(B)}$. For any $b\in B$, let $b\in C$ where C is a maximal Boolean subalgebra of $B$. Then $b=\bigvee A$, $A\subseteq A(C)\subseteq A(B)$ due to the lemma in the proof of theorem \ref{thm1}. Thus $p_1(b)=p_1(\bigvee A)=\sum_{a\in A}p_1(a)=\sum_{a\in A}p_1|_{A(B)}(a)=\sum_{a\in A}p_2|_{A(B)}(a)=\sum_{a\in A}p_2(a)=p_2(\bigvee A)=p_2(b)$, so $p_1=p_2$, which induces a contradiction. Therefore $f$ is injective.\par

If $p'\in s(AG(B))$, then define $p:B\to[0,1]$ as below.
\begin{equation}
\begin{split}
p:B&\to [0,1]\\
0&\mapsto 0\\
b=\bigvee A&\mapsto\sum_{a\in A}p'(a),\ (A\subseteq A(B))
\nonumber
\end{split}
\end{equation}

Then $p|_{A(B)}=p'$, we prove that $p\in s(B)$. $p(b)=\sum_{a\in A}p'(a)\in[0,1]$ because $p'\in s(AG(B))$ and $A$ is contained in a maximal Boolean subalgebra. If $b=\bigvee A_1=\bigvee A_2$, then $A'_1\cup A_2=A(D_1)$ and $A'_1\cap A_2=\emptyset$ due to the proof of lemma in the proof of theorem \ref{thm1}, so $p(b)=\sum_{a\in A_1}p'(a)=1-\sum_{a\in A'_1}p'(a)=\sum_{a\in A_2}p'(a)$. Therefore, $p$ is well-defined.\par

We have $p(0)=0$. If $b=\bigvee A$, $p(\neg b)=p(\bigvee A')=\sum_{a\in A'}p'(a)=1-\sum_{a\in A}p'(a)=1-p(b)$. If $x,y\in B$ and $x\odot y$, then $x,\ y$ are in the same maximal Boolean subalgebra $C$. Let $x=\bigvee A_x,\ y=\bigvee A_y$ where $A_x, A_y\subseteq A(C)$. We have $p(x\lor y)+p(x\land y)=\sum_{a\in A_x\cup A_y}p'(a)+\sum_{a\in A_x\cap A_y}p'(a)=\sum_{a\in A_x}p'(a)+\sum_{a\in A_y}p'(a)=p(x)+p(y)$, so $p\in s(B)$. Therefore, $f(p)=p'$. $f$ is surjective.\par

In conclusion, $f$ is a bijection between $s(B)$ and $s(AG(B))$.
\end{proof}

The theorems \ref{thm2} shows the one-to-one correspondence between states on an finite dimensional $epBA$ and states on its atom graph. If $Q\in QS$, then a quantum state $\rho$ induces a map $\rho: Q\to[0,1]$, $\rho(\hat{P})=tr(\rho\hat{P})$. It is easy to prove that $\rho$ is a state on $Q$. Let $qs(Q)$ denote the states on $Q$ induced by quantum states. We have $qs(Q)\subseteq s(Q)\cong s(AG(Q))$ for finite dimensional $QS$.

Theorems \ref{thm1} and \ref{thm2} expose the graph structure of finite dimensional $epBA$, thus prove the reasonability of graphs to be the models of finite dimensional quantum systems. Firstly, the measurement scenario is determined by atom graph. And then, the quantum states are determined by the states on atom graph.

\section{Extension of graphs to atom graphs}

The section below explains how the atom graphs are connected to the exclusivity graphs.

At first, we introduce concept of exclusivity graphs, which is the application of graph theory for quantum contextuality.

\subsection{Exclusivity graph}
Exclusivity graphs \citep{Adan2010Contextuality,Adan2014Graph} are the tools utilized to describe the exclusive events. It is based upon the mathematical works of Lov\'{a}sz et al. \citep{Lovasz1986Relaxations}.

Let $G$ be a finite graph. The vertexes of $G$ are marked as $1, 2, ..., n$. A vector $x:\ V(G)\rightarrow \{0,\ 1\}$ in $\{0\ ,1\}^n$ is said to be the incidence vector of vertex set $x^{-1}(1)\subseteq V(G)$.

Notation $\alpha(G;w)$ denotes the maximum weight of the independent sets of $G$. A weight is a vector $w:\ V(G)\rightarrow\mathbb{R}^+$. Thus $\alpha(G;\vec{1})$ is the maximum independent number of $G$, also written $\alpha(G)$.

 Let $VP(G)$ (vertex packing polytope) indicate the convex hull of incidence vectors of all the independent sets of vertexes. $VP(G)$ was employed in the calculation of $\alpha(G;w)$, because $\alpha(G;w)$ is the maximum of the linear function $w^Tx$ for $x\in VP(G)$. Moreover, $VP(G)$ consists of all of the ``classical probabilities" from the perspective of exclusivity graphs.

A well-known example is the KCBS experiment \citep{Alexander2008Simple}, which includes five rank-1 projectors $\hat{P}_0,\ \hat{P}_1,\ \hat{P}_2,\ \hat{P}_3,\ \hat{P}_4$ in 3-dimensional Hilbert space such that $\hat{P}_i$ and $\hat{P}_{i+1}$  (with the sum modulo 5) are orthogonal, that is, exclusive. The exclusivity relation of $\hat{P}_i\ (i=0,\ 1,\ 2,\ 3,\ 4)$ is shown by Fig.\ref{fig8}.
\begin{figure}[H]
    \centering
    \includegraphics[width=0.3\linewidth]{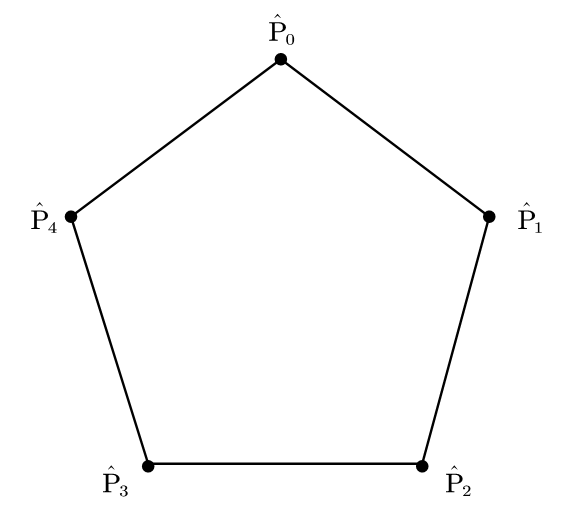}
    \caption{: Exclusivity graph $G$ for KCBS experiment}\label{fig8}
 \end{figure}

 If these five vertexes are ``classical events", that is, the events in classical probability theory (sets or indicator functions), for the graph $G$, $VP(G)$ is the set of classical probability vectors of five events. In that case, polynomial $\sum_{i=0}^4p(P_i)$ satisfies the inequality below.

$$\sum_{i=0}^4p(\hat{P}_i)\leq\alpha(G)=2$$

It is called the KCBS inequality, the earliest NC inequality \citep{Alexander2008Simple}.

However, in quantum case, these five events are ``quantum events". If the quantum state is $\rho$, then the probabilities of event $P_i$ is $\langle\hat{P}_i\rangle=Tr(\rho\hat{P}_i)$. A notable interpretation of $G$ in quantum systems was found by Cabello et al. \citep{Adan2010Contextuality}. It violates the KCBS inequality,

$$\sum_{i=0}^4p(\hat{P}_i)=\sum_{i=0}^4 Tr(\rho\hat{P}_i)=\sqrt{5}>2$$

which provides an evidence of quantum contextuality.

An exclusivity graph isn't necessarily an atom graph, but we can prove that every finite graph is the induced graph of an atom graph, which leads to the connection between finite dimensional $epBA$ and exclusivity graphs.

\subsection{Faithful and linearly independent orthogonal co-representation }
To achieve our final conclusion, it's necessary to introduce the notion orthogonal co-representation.

Gr\"{o}tschel, Lov\'{a}sz, and Schrijver (1986) defined the Orthonormal Representation (OR) of graph $G$, which can be seen as an interpretation of $G$ to quantum systems \citep{Lovasz1986Relaxations}.

\begin{definition}
Let $G$ be a graph. An \textbf{OR} of $G$ is a map $v:V(G)\rightarrow\mathbb{R}^d$ $(d\in\mathbb{Z}^+)$ such that $\parallel v(i)\parallel=1(i\in V(G))$, and if $i,\ j$ are not adjacent, then $v(i)\bot v(j)$.

Let $v_i$ denote vector $v(i)$ in the following.
\end{definition}

Notice that if $i,j$ are not adjacent then corresponding vectors are orthogonal. To interpret adjacency to orthogonality, we should consider the OR of $\bar{G}$ (complement of $G$). Thus, Abramsky and Brandenburger proposed the faithful orthogonal co-representation\citep{Abramsky2011sheaf}.

\begin{definition}
Let $G$ be a graph, and $v$ an OR of $\bar{G}$. $v$ is said to be a \textbf{faithful orthogonal co-representation} of $G$ if $v$ is injective and $i,j$ are adjacent iff $v_i\bot v_j$.
\end{definition}

This definition ensures that a graph corresponds to unique orthogonality graph. Furthermore, we define another notion.

\begin{definition}
Let $G$ be a graph, and $v$ an OR of $\bar{G}$. $v$ is said to be a \textbf{linearly independent orthogonal co-representation} of $G$ if vector set $\{v_i:i\in V(G)\}$ are linearly independent.
\end{definition}

For graph $G$, it is important to know whether it has an orthogonal co-representation or not. A result for the question was mentioned in \cite{Abramsky2011sheaf}, but it is not powerful enough for us. Here, we need to prove another result.

\begin{theorem}\label{thm4}
Each finite graph $G$ with $n$ vertexes has a faithful and linearly independent orthogonal co-representation in $\mathbb{R}^n$
\end{theorem}
\begin{proof}
    We use the mathematical induction. When $n=1$, $G$ has a faithful and linearly independent orthogonal co-representation in $\mathbb{R}$, that is, $\{v_1\}$. Then we assume the result for general $n-1$, and show it holds for $n$.

    For a graph $G$ with $n$ vertexes, its every induced subgraph with $n-1$ vertexes has a faithful and linearly independent orthogonal co-representation: $\{v_1,v_2,...,v_{n-1}\}$. They span an $(n-1)$-dimensional subspace of $\mathbb{R}^n$. The problem is thus reduced to proving that there is a vector $v_n\in\mathbb{R}^n$ such that $\parallel v_n\parallel=1$, $v_n\bot v_i$ iff $n,i$ are adjacent $(i=1,2,...,n-1)$, and $\{v_1,...,v_{n-1},v_n\}$ are linearly independent.

    The subspace $Span(v_1,...,v_{n-1})^{\bot}$ is one-dimensional. Let $e_n$ be a unit vector in it, and then $\{v_1,...,v_{n-1},e_n\}$ is a basis of $\mathbb{R}^n$. Suppose $v'_n=x_1v_1+...+x_{n-1}v_{n-1}+e_n$. We need $v'_n$ satisfying:

     1.$x_1(v_i,v_1)+...+x_{n-1}(v_i,v_{n-1})+0=0$, i.e. $v'_n\bot v_i$, iff $n,\ i$ are adjacent;

     2.$x_1(v_i,v_1)+...+x_{n-1}(v_i,v_{n-1})+0\neq0$ iff $n,i$ are not adjacent,

     where $(,)$ is the notation for inner product on $\mathbb{R}^n$.

    If $n$ is adjacent with all of the $i=1,...,n-1$, we have
\begin{equation}
    \begin{cases}
     x_1(v_1,v_1)+x_2(v_1,v_2)+...+x_{n-1}(v_1,v_{n-1})=0\\
     x_1(v_2,v_1)+x_2(v_2,v_2)+...+x_{n-1}(v_2,v_{n-1})=0\\
     ......\\
     x_1(v_{n-1},v_1)+x_2(v_{n-1},v_2)+...+x_{n-1}(v_{n-1},v_{n-1})=0.
    \end{cases}\nonumber
\end{equation}

Because $\{v_1,...,v_{n-1}\}$ is linearly independent, $((v_i,v_j))$ is a Gram matrix, with $rank \ n-1$. Thus the equation system has a unique solution $\vec{x}=0$, which gives $v'_n=e_n$.

If $n$ is not adjacent with some vertexes, then we substitute the corresponding equations with inequalities in the equation system. It can be proved that the new system still has solutions.

If the system has $m$ equalities, then the subsystem made up of them has an $(n-1-m)$-dimensional solution space $S$. We mark the $n-1-m$ inequalities left as $1,2,...,n-1-m$, which respectively have $(n-2)$-dimensional solution spaces $S_1,S_2,...,S_{n-1-m}$ as equalities. We have

\begin{equation}
\begin{split}
S'&=\{\vec{x}:\vec{x}\in S\ and\ \vec{x}\notin S_1,...,S_{n-1-m}\}\\
   &=S\cap\overline{S_1}\cap...\cap\overline{S_{n-1-m}}\\
   &=S\cap(\overline{S}\cup\overline{S_1})\cap...\cap(\overline{S}\cup\overline{S_{n-1-m}})\\
   &=S\cap(\overline{(S\cap S_1)}\cap...\cap\overline{(S\cap S_{n-1-m})})\\
   &=S\cap\overline{(S\cap S_1)\cup...\cup(S\cap S_{n-1-m})}\\
   &=S-((S\cap S_1)\cup...\cup(S\cap S_{n-1-m})).
\end{split}\nonumber
\end{equation}

$S\cap S_i$ are all $(n-m-2)$-dimensional subspaces of $(n-m-1)$-dimensional space $S$, that is, hyperplanes. Since the union of finite hyperplanes properly contains in the whole space, $S'$ is not empty. Thus the new system has solutions.

Therefore, there exists a vector $v'_n=x_1v_1+...+x_{n-1}v_{n-1}+e_n$ such that $v'_n\bot v_i$ iff $n,i$ are adjacent. Since $e_n$ has coefficient 1 in $v'_n$, $\{v_1,...,v_{n-1},v'_n\}$ is linearly independent. Finally, let $v_n=v'_n/\parallel v'_n\parallel$. Then $v_n$ is the desired vector, so $G$ has a faithful and linearly independent orthogonal co-representation in $\mathbb{R}^n$, and the induction goes through.
\end{proof}

Due to theorem \ref{thm4}, the pentagon in Fig.\ref{fig8} has a faithful and linearly independent orthogonal co-representation in $\mathbb{R}^5$, which differs from the one found by Cabello et al. in $\mathbb{R}^3$ \citep{Adan2010Contextuality}. Both of them can be used to investigate the probabilities of exclusive events, and one will see that the linear independence has special benefits.

\subsection{Higher dimensional context extension}
With theorem \ref{thm4}, we show how to extend a graph to an atom graph.

 It is easy to see there is a one-to-one correspondence with finite graph $G$ and its total maximal cliques. Therefore we can define

\begin{definition}
 Let $G$ be a finite graph, whose total maximal cliques are $C_1,C_2,...,C_N$. Let $x_1,...,x_N$ be $N$ vertexes irrelevant to vertexes in $G$. The \textbf{higher-dimensional context extension} of $G$, denoted by $G^e$, is a graph with total maximal cliques: $C_1\cup\{x_1\},C_2\cup\{x_2\},...,C_N\cup\{x_N\}$.
\end{definition}

We call the size of the maximum clique of $G$ \textbf{dimension} of $G$. $G^e$ is gotten by adding a vertex into every maximal clique of $G$. The dimension of $G^e$ must be higher than $G$. It is a tool to study $G$ in higher dimensions.

For convenience, we introduce the notion "substate" for subgraph.

\begin{definition}
A \textbf{substate} on a graph $G$ is defined by a map $p:V(G)\rightarrow[0,\ 1]$ such that for each maximal clique $C$ of $G$, $\sum_{v\in C}p(v)\leq1$. Use $ss(G)$ to denote the set of substates on $G$.
\end{definition}

And we have

\begin{proposition}\label{prop5}
    If $G$ is a finite graph, then $ss(G)\cong s(G^e)$
\end{proposition}
\begin{proof}
    For each state on $G^e$, its restriction on $G$ is a substate from the definition.

    For each $v\in ss(G)$, if the total maximal cliques of $G^e$ are $C_1\cup\{x_1\},C_2\cup\{x_2\},...,C_N\cup\{x_N\}$, then we define a state $v'$ on $G^e$ by $v'(i)=v(i)(i\in V(G))$ and $v'(x_k)=1-\sum_{i\in C_k}v(i)$. It is easy to see that $v'$ is the unique state on $G^e$ such that $v'(i)=v(i)(i\in V(G))$.
\end{proof}

Now we prove the central theorem of this section.

\begin{theorem}\label{thm6}
    If $G$ is a finite graph, then $G^e$ is the atom graph of an finite $epBA$.
\end{theorem}
\begin{proof}
      Suppose $|V(G)|=n$. Applying theorem \ref{thm4}, $G$ has a faithful and linearly independent orthogonal co-representation $\{v_1,...,v_n\}$ in $\mathbb{R}^n$, which correspond to a projector set $\{\hat{P}_1,...,\hat{P}_n\}$, where $\hat{P}_i$ projects $\mathbb{R}^n$ to $Span(v_i)$.

     If $G$ has $N$ maximal cliques $C_1,...,C_N$, we set $P_k=\lnot(\bigvee_{i\in C_k}\hat{P}_i)(k=1,...,N)$. $P_k$ is a projector onto $Span(\{v_i:i\in C_k\})^{\bot}$. It will be proved that $\{\hat{P}_1,...,\hat{P}_n\}\cup\{P_1,...,P_N\}$ is the set of atoms of an $epBA$, whose atom graph is exactly $G^e$.

    Firstly we prove that every $P_k$ is commeasurable with $\hat{P}_i(i\in C_k)$, but not commeasurable with $\hat{P}_j(j\notin C_k)$. If $P_k$ is commeasurable with a $\hat{P}_j(j\notin C_k)$, thus $P_k\hat{P}_j=\hat{P}_jP_k$, then we have $P_k\hat{P}_j=\hat{P}_j$ or $P_k\hat{P}_j=\hat{0}$.

    If $P_k\hat{P}_j=\hat{P}_j$, then $v_j\in Span(\{v_i:i\in C_k\})^{\bot}$, so $v_j\bot v_i$ for all $i\in C_k$. Thus $C_k\cup{v_j}$ is a clique, which contradicts to that $C_k$ is maximal.

    If $P_k\hat{P}_j=\hat{0}$, then $v_j\bot Span(\{v_i:i\in C_k\})^{\bot}$, so $v_j\in Span(\{v_i:i\in C_k\})$, which contradicts to that $\{v_1,...,v_n\}$ is linearly independent.

    Next we prove that each $P_k(k=1,...,N)$ is not commeasurable with others. It is trivial for $N=1$. If $N>1$, and $P_{k_1}$, $P_{k_2}$ ($k_1\neq k_2$) are commeasurable, we have $P_{k_1}=P'\vee P$, $P_{k_2}=P''\vee P$ and $P'P''=P'P=P''P=0$, $P=P_{k_1}P_{k_2}$.

    For maximal cliques $C_{k_1}$ and $C_{k_2}$, set $A_1=C_{k_1}-C_{k_2}$, $A_2=C_{k_1}\cap C_{k_2}$ and $A_3=C_{k_2}-C_{k_1}$. $A_1$ and $A_3$ are non-empty. Let $P_{A_t}=\bigvee_{i\in A_t}\hat{P}_i(t=1,2,3)$, which gives that $P'\vee P\vee P_{A_1}\vee P_{A_2}=P_{k_1}\vee\bigvee_{i\in C_{k_1}}\hat{P}_i=\hat{1}$, and $P''\vee P\vee P_{A_2}\vee P_{A_3}=P_{k_2}\vee\bigvee_{i\in C_{k_2}}\hat{P}_i=\hat{1}$, as shown in Fig.\ref{fig11}.
    \begin{figure}[H]
    \centering
    \includegraphics[width=0.3\linewidth]{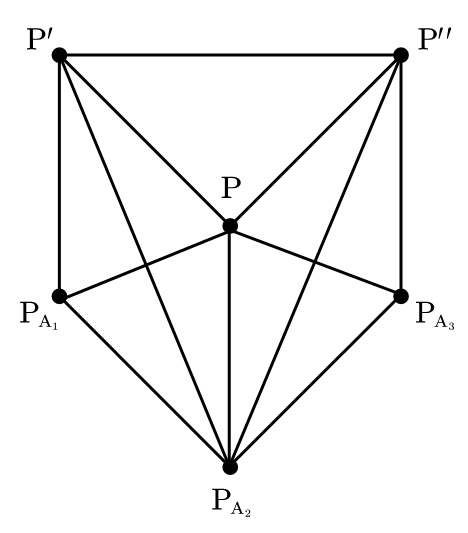}
    \caption{The orthogonal graph of $P$, $P'$, $P''$, $P_{A_1}$, $P_{A_2}$ and $P_{A_3}$.}\label{fig11}
    \end{figure}

    $P',\ P''$ and $P$ are pairwise orthogonal, so they generate a Boolean algebra. Thus $P_{A_1}\vee P_{A_2}=\lnot(P'\vee P)$ and $P_{A_2}\vee P_{A_3}=\lnot(P''\vee P)$ are in the same Boolean algebra, i.e. $P_{A_1}\vee P_{A_2}$ and $P_{A_2}\vee P_{A_3}$ are commeasurable. Then there exist two projectors $P'_{A_1}$ and $P'_{A_3}$ such that $P'_{A_1}P'_{A_3}=\hat{0}$, $P_{A_1}=P'_{A_1}\vee P_{A_1}P_{A_3}$ and $P_{A_3}=P'_{A_3}\vee P_{A_1}P_{A_3}$, . Thus $P_{A_1}$ and $P_{A_3}$ are commeasurable.

    Because $\{v_i:i\in A_1\cup A_3\}$ is linearly independent, if $x=\sum_{i\in A_1}a_iv_i=\sum_{j\in A_3}b_jv_j$, then $\sum_{i\in A_1}a_iv_i-\sum_{j\in A_3}b_jv_j=0$. Thus $a_i,b_j=0$, i.e. $x=\vec{0}$, which means that $P_{A_1}P_{A_3}=\hat{0}$, in other words, $Span(\{v_i:i\in A_1\})$ and $Span(\{v_i:i\in A_3\})$ are orthogonal. However, it leads that the vectors in $\{v_i:i\in A_1\}$ and $\{v_i:i\in A_3\}$ are orthogonal. Since $v$ is faithful, we have that $C_{k_1}\cup C_{k_2}$ is a clique. It causes $C_{k_1}=C_{k_2}$, which contradicts to $k_1\neq k_2$.

    To sum up, $P_k$ is only commeasurable with $\hat{P}_i(i\in C_k)$. Let $B$ be a partial Boolean algebra generated by the set $A=\{\hat{P}_1,...,\hat{P}_n\}\cup\{P_1,...,P_N\}$. For each $i\in\{1,...,n\}$, $\hat{P}_i$ is an atom obviously. For each $k\in\{1,...,N\}$,  because $P_k$ is only commeasurable with $\hat{P}_i(i\in C_k)$, and $P_k\hat{P}_i=\hat{0}(i\in C_k)$, the Boolean algebra generated by $\{P_k\}\cup\{\hat{P}_i:i\in C_k\}$ is isomorphic to the Boolean algebra with $2^{|C_k|+1}$ elements. Then we can see $P_k$ is an atom. And $A$ contains all the atoms of $B$. $B$ is a partial subalgebra of projectors, so it is an $epBA$, and the atom graph of $B$ is isomorphic to $G^e$.
\end{proof}

With theorem \ref{thm6}, we immediately get that each finite graph is an induced subgraph of the atom graph of an finite $epBA$.

For instance, the exclusivity graph $G$ in Fig.\ref{fig3} is an induced subgraph of the atom graph in Fig.\ref{fig12}, where $\hat{P}_{ij}:=\neg(\hat{P}_i\lor\hat{P}_j)$

\begin{figure}[H]
    \centering
    \includegraphics[width=0.4\linewidth]{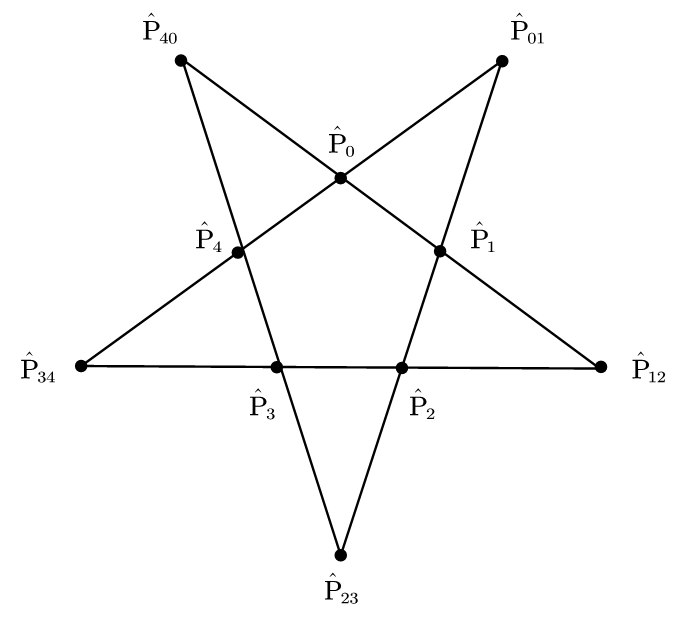}
    \caption{The atom graph of KCBS system}\label{fig12}
\end{figure}

$G^e$ complements all the elementary events overlooked by $G$. Another approach to get all the elementary events was used by Cabello et al.\citep{Adan2014Graph}, which lets Fig.\ref{fig8} be a subgraph of the graph in Fig.\ref{fig13}.

\begin{figure}[H]
    \centering
    \includegraphics[width=0.5\linewidth]{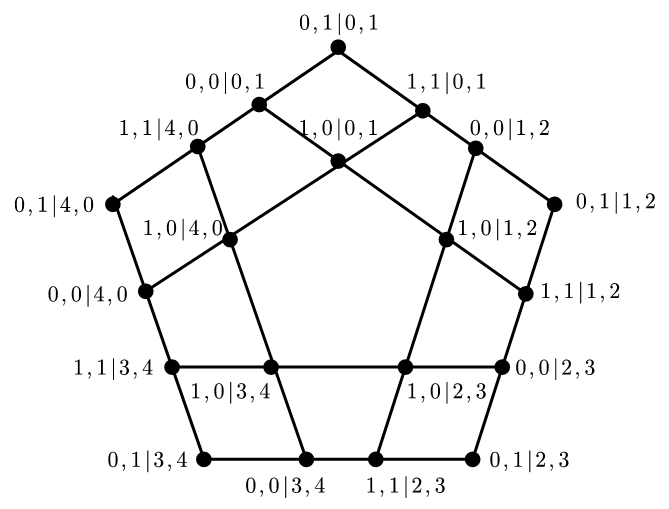}
    \caption{The exclusivity graph for KCBS experiment used by \cite{Adan2014Graph}. $x,y|i,j$ represents that the outcomes $\hat{P}_i,\hat{P}_j$ are $x,y$ $(i,j=0,1,2,3,4.\ x,y=0,1)$.}\label{fig13}
\end{figure}

However, because of the exclusivity relation of $\hat{P}_i(i=0,1,2,3,4)$, the vertexes $1,1|i,i+1$ in Fig.\ref{fig13} are all impossible events, which should be deleted. $1,0|i,i+1$ and $0,1|i-1,i$ are equivalent, so they should be merged in pairs. Therefore, Fig.\ref{fig13} doesn't give a correct expression of the KCBS experiment. After simplifying, Fig.\ref{fig13} is reduced to Fig.\ref{fig12}, which presents all the elementary events precisely.

The higher-dimensional context extension of graph $G$ is one way to extend $G$. Another method is the equal-dimensional context extension, which adds one point to every maximal clique of $G$ except the maximum cliques. However, different from $G^e$, the equal-dimensional context extension may not be an atom graph, which means that $G$ may have no interpretation to an equal-dimensional quantum system.

For now, we have connected the atom graph and the exclusivity graph. The next section moves on to consider the KS contextuality.

\section{KS contextuality}
KS theorem is the earliest description of the contextuality of quantum systems \citep{Kochen1967The}. It can be depicted by partial Boolean algebra \citep{Abramsky2020The}, which states that if $\mathcal{H}$ be a Hilbert space and $dim(\mathcal{H})\geq 3$, then there is no homomorphism from $\mathcal{P}(\mathcal{H})$ to $\{0,1\}$ (Kochen and Specker, 1967).

In other words, when $dim(\mathcal{H})\geq 3$, it is impossible to assign truth-values to all properties in quantum systems simultaneously, which leads to the impossibility of assigning values to all observables simultaneously. The property ``there is no homomorphism to $\{0,1\}$" is called KS contextuality.

With theorems \ref{thm1} and \ref{thm2}, we have that

\begin{proposition}\label{prop7}
    If $Q$ is a finite dimensional $QS$, then $Q$ presents KS contextuality iff there is no 0-1 state on $AG(Q)$.
\end{proposition}

Two important KS-proofs with graphs were given by Kochen, Specker \citep{Kochen1967The}, and Cabello et al \citep{Cabello1997Bell}. The KS graph has 117 vertexes, and Cabello's graph has 18 vertexes. If the existence of 0-1 states on a graph leads to a contradiction, we say it introduces a KS-proof. The remaining part of this section will offer a general and parametric expression of KS contextuality, which introduces a type of NC inequality.

\begin{definition}
     If $G$ is a finite graph, and $i\in V(G)$. define $c_G(i):=||\{C|C\ is\ a\ maximal\ clique\ of\ G,\ i\in C|\}$. If $S\subseteq V(G)$, define $c_G(S):=\sum_{i\in S}c_G(i)$
\end{definition}

We call $c_G(i)$ the number of associated contexts of $i$, and if $I\in V(G)$ is an independent set of $G$, $c_G(I)$ is called a number of independently associated contexts of $G$. Therefore, $\alpha(G;c_G)$ is the greatest number of independently associated contexts of $G$. Exactly, what we define is a special weight $c_G$.

\begin{definition}
If $G$ is a finite graph, and $v\in s(G)$, define $S(v,c_G):=\sum_{i\in V(G)}c_G(i)v(i)$, and $c(G)$ denotes the total number of maximal cliques of $G$.
\end{definition}

It is straightforward to show that $S(v,c_G)=c(G)$ by $S(v,c_G)=\sum_{k=1}^N\sum_{i\in C_k}v(i)=\sum_{k=1}^N1=N$, where $C_k$ is the k'th maximal clique of $G$. $S(v,c_G)=c(G)$ is an important equation for states on quantum systems, and it is also a generalization of the equation used for KS-proof in \cite{Cabello2009Universality}.

\begin{lemma}\label{prop8}
   If $G$ is a finite graph, then $\alpha(G;c_G)\leq c(G)$.
\end{lemma}
\begin{proof}
If $I$ is an independent set of $G$, two vertexes $i,j\in I$ can not associate to the same maximal clique. Otherwise, if $i,j\in C$, then $i,j$ are adjacent since $C$ is a clique, which contradicts to that $I$ is an independent set. Therefore, distinct vertexes in $I$ associate to distinct maximal cliques. Thus $c_G(I)\leq c(G)$ for any independent set $I$ of $G$, which deduces that $\alpha(G;c_G)\leq c(G)$.
\end{proof}

Next, we give description of KS contextuality using the parameters of graphs.

\begin{theorem}\label{thm9}
   If $G$ is a finite graph, then the statements below are equivalent:\\
    $1.\ \alpha(G;c_G)=c(G)$.\\
    $2$. There exists a 0-1 state on $G$. \\
    $3$. There exists a 0-1 state $v\ on\ G\ s.t.\ S(v,c_G)=\alpha(G;c_G)$.\\
    $4$. There exists a state $v\ on\ G\ s.t.\ S(v,c_G)=\alpha(G;c_G)$.
\end{theorem}
\begin{proof}
$1\Rightarrow 2$: Since $\alpha(G;c_G)=c(G)$, there is an independent set $I$ satisfying $c_G(I)=c(G)$. Thus vertexes in $I$ associate to all the maximal cliques of $G$. We define a map $v:V(G)\rightarrow \{0,1\}$ by $v(i)=1,(i\in I)$ and $v(i)=0,(i\notin I)$. Then $v$ is a 0-1 state on $G$.

$2\Rightarrow 3$: If $v$ is a 0-1 state on $G$, then the set $I=\{i\in V|v(i)=1\}$ is an independent set. Since for every maximal clique $C$, $\sum_{i\in C}v(i)=1$, there is exactly one vertex $i$ such that $v(i)=1$ in $C$. Thus the vertexes in $I$ associate to all the maximal cliques. Therefore, $c_G(I)=c(G)\leq\alpha(G;c_G)$. Applying lemma \ref{prop8}, we have $c(G)=\alpha(G;c_G)$. Therefore $S(v,c_G)=\alpha(G;c_G)$

$3\Rightarrow 4$: Follows from the relevant definitions.

$4\Rightarrow 1$: Obviously from the equation $S(v,c_G)=c(G)$.
\end{proof}

Notice that the theorem \ref{thm9} holds for the atom graph of any finite $epBA$. Therefore, applying proposition \ref{prop7}, we have that a finite quantum system $Q$ presents KS contextuality iff $S(v,c_{AG(Q)})=c(AG(Q))>\alpha(AG(Q);c_{AG(Q)})$ for all $v\in s(AG(Q))$. It supplies a parametric method, also a NC inequality, to determine if the finite quantum system has KS contextuality. A similar result was gotten with sheaf theory by Abramsky and Brandenburger \citep{Abramsky2011sheaf}. However, their expression is not parametric compared with ours, and they didn't realize the graph structure of quantum systems.

\section{Conclusion}

We exposed the graph structure of finite dimensional quantum systems by theorems \ref{thm1} and \ref{thm2} for $epBA$, which ensures that the utilization of graphs for quantum systems is reasonable. $epBA$, with the atom graphs we defined, can be used to describe the finite dimensional quantum systems and develop the theories for quantum contextuality. As an instance, a general and parametric description of KS contextuality for finite quantum systems was presented by the theorem \ref{thm9}.

In the rest of this paper, we establish the connection between atom graph and exclusivity graph by theorems \ref{thm4} and \ref{thm6}, which introduces a method to express the exclusivity experiments more precisely. The higher-dimensional (or equal-dimensional) context extension can be tools to investigate the features of quantum experiments.

\bmhead{Acknowledgments}
The work was supported by National Natural Science Foundation of China (Grant No. 12371016, 11871083) and National Key R\&D Program of China (Grant No. 2020YFE0204200).

\bibliography{sn-bibliography}

\end{document}